\documentclass[journal]{IEEEtran}

\usepackage{graphicx} 
\usepackage[noadjust]{cite}                               
\usepackage{amsmath}                    
\usepackage{amssymb}
\usepackage{amsthm}
\usepackage{subcaption}

\newtheorem{theorem}{Theorem}
\newtheorem{remark}{Remark}

\newtheorem{lemma}{Lemma}
\newtheoremstyle{noparens}%
  {}{}%
  {\itshape}{}%
  {\bfseries}{:}
  { }%
  {\thmname{#1}\thmnumber{ #2}\mdseries\thmnote{ #3}}

\theoremstyle{noparens}
\newtheorem{lemmaNoParens}[lemma]{Lemma}

\usepackage{multicol}
\usepackage{multirow}
\usepackage[english]{babel}   
\usepackage{color}
\usepackage[T1]{fontenc}
\usepackage{soul}
\usepackage{epstopdf}

\DeclareMathOperator{\sat}{sat}
\DeclareMathOperator{\sign}{sign}
\DeclareMathOperator{\sig}{sig}

\title{Boundary-layer control with unstructured uncertainties with application to adaptive autopilots}
\author{Peng Li, Di Liu~\IEEEmembership{Member,~IEEE}, 
and Simone Baldi~\IEEEmembership{Senior member,~IEEE}\vspace{-0.2cm} 
\thanks{This research was partly supported by the National Key R\&D Program of China grant 2022YFE0198700, by the EU Horizon 2020 R\&I programme Marie Sklodowska-Curie grant 899987, and by the Natural Science Foundation of China grant 62073074 \emph{(corresponding authors: S. Baldi and D. Liu)}
\newline
  P. Li is with School of Cyber Science and Engineering, Southeast University, China (email: {\tt\small lpeng\_2013@163.com}) \newline  
D. Liu is with School of Computation, Information and Technology, Technical University of Munich, Germany 
(e-mail: di.liu@tum.de).
\newline  
S. Baldi is with Frontier Center for Mobile Information Communication and Security, Southeast University, China (email: {\tt\small s.baldi@tudelft.nl})  
}
}

\begin{document}
\maketitle

\begin{abstract}
    Control with unstructured uncertainties refers to controlling systems where not only the parameters are unknown, but also the way the parameters appear in the dynamics. This problem {\color{black}becomes pivotal in autopilots, where a unified control architecture is sought 
    for aerial/ground/marine vehicles with different structures}. 
    {\color{black}By only making use of basic Euler-Lagrange properties valid in most mechanical systems independently of their specific structure, this brief proposes an adaptive design that does not rely on structural knowledge of the uncertainties. The proposed adaptive method, here validated in the ArduPlane module of ArduPilot, applies also to other modules like ArduCopter, ArduRover, ArduSub.} Enhanced performance with respect to state-of-the-art methods addressing unstructured and {\color{black}state-dependent} uncertainties is verified.\vspace{-0.1cm} 
\end{abstract}

\section{Introduction}
{\color{black}The most advanced autopilot suites aim at unified control architectures that can tackle unmanned vehicles with different structures, no matter if the vehicle is}  {\color{black}aerial} \cite{8693689,9703092,9589106}, ground \cite{9629337, VONELLEN102750}, 
 or marine  \cite{2006Adaptive} (cf. the ArduPlane, ArduCopter, ArduRover, ArduSub modules in ArduPilot, or similar modules in other autopilots).   {\color{black}Such autopilot suites are the paradigmatic illustration of control with unstructured uncertainties}, since {\color{black} parametric perturbations and unmodeled dynamics \cite{8438315,8304792} are not only stemming from measurement noises and environmental disturbances \cite{9570132, 9325954}, but also  from the different structures of aerial/ground/marine vehicles.} 
 {\color{black}Imposing structural assumptions is not viable}  {\color{black}since  mismatches between a priori assumptions on the uncertainty and reality may fool the control law, degrade performance, and lose stability \cite{https://doi.org/10.1002/rnc.3562}.}  {\color{black} Avoiding structural assumptions plays a key role in such autopilots.}
 
 {\color{black}As unmanned vehicles are mechanical systems, a convenient perspective to autopilot design is the control of uncertain Euler-Lagrange dynamics.} 
These control methods can be roughly categorized according to robust methods and adaptive methods. Common robust methods, based on  observers \cite{9369849}, passivity \cite{7859476}, or sliding mode control (SMC) \cite{9570130, 9512516}, eventually rely on some prior knowledge of the uncertainty, making it difficult to cope with uncertainties going beyond the expected bounds during system operation \cite{https://doi.org/10.1002/rnc.3562}. 
Even adaptive methods are not free of prior knowledge of the uncertainty, appearing in the form of structural conditions (e.g. linear-in-parameter, matching conditions \cite{8435993,9721593}), or of constant bounds for the uncertainty and its time derivative \cite{9731672,doi:10.1080/00207179.2010.501385,2016Adaptive,2020Adaptive}. 

{\color{black}As structural uncertainty is intrinsically state-dependent,} 
 imposing constant bounds restrictively amounts to assuming 
bounded states before proving stability \cite{2020Ona, 9715173}. Recent years have seen progress in the control of Euler-Lagrange dynamics with reduced structural knowledge, but  
{\color{black} not all these methods can be seamlessly embedded in autopilots to fit the needs of their control architecture: e.g., several structure-free SMC methods do not allow integral action needed in autopilot loops \cite{9001188, 9525201,9795675}, while structure-free methods based on prescribed performance or funnel functions \cite{9030217,9165001} pose the problem of large inputs.} {\color{black}Finite-time stability is sought in autopilots to enhance robustness, but available finite-time methods typically rely on  structural conditions on the uncertainty 
\cite{7331625, 9492051, 9271817,doi:10.1080/00207179.2016.1138241}.}

{\color{black}As a result, a systematic structure-free way to enhance robustness and adaptation of unified autopilot architectures is a notable open issue. This work contributes in this sense by:} 
\begin{itemize}
    \item[a)] Proposing a novel sliding surface that embeds a non-singular finite-time term and an integral term, useful to incorporate control law into existing autopilot loops. 
    \item[b)] Removing structural knowledge of the uncertainty. 
   Only basic Euler-Lagrange properties are used, valid for  mechanical systems independently on their specific structure. 
    \item[c)] 
    Proposing an adaptive law in the framework of finite-time control. The law tackles unstructured uncertainty by estimating a state-dependent bound of the uncertainty, whose {\color{black}form is independent on the specific system structure}.  
\end{itemize}

Section II formulates the control with unstructured uncertainty. Section III covers the design and analysis. {\color{black}Software-in-the-loop} validations are in Section IV, with comparisons to the existing 
ArduPilot 
and to an adaptive SMC autopilot. {\color{black}Additional comparisons are in the report \cite{TechReportandCode}.} 

\textbf{Notations}: Let $\mathbb{R}$ and $\mathbb{R}^{n\times m}$ be the sets of real numbers and real $n \times m$ matrices, and $\mathbb I_n$ be the  $n\times n$ identity matrix. Vectors are denoted with bold, such as $\mathbf{x} = [{x}_1, \cdots, {x}_n]^T$. {\color{black}Let us use the short notations $\sign(\mathbf{x})=[\sign(x_1), \cdots, \sign(x_n)]^T$,} 
$\sig^v(\mathbf{x})=[|x_1|^v \sign(x_1), \cdots, |x_n|^v \sign(x_n)]^T$. Let $\|\cdot\|$ and $\lambda_{min}(\cdot)$ denote the Euclidean norm and minimum eigenvalue.

\section{Preliminaries and Problem Formulation}
Consider the following Euler-Lagrange (EL) dynamics
\begin{align}
    \mathbf M( \mathbf q)\ddot{ \mathbf q}+ \mathbf C(\mathbf q,\dot{\mathbf q})\dot{\mathbf q}+ \mathbf G(\mathbf q)+  \mathbf F(\dot{\mathbf q})+ \mathbf d = \mathbf u, \label{EL}
\end{align}
with $\mathbf q, \dot{\mathbf{q}}, \ddot{\mathbf{q}} \in \mathbb{R}^n$ denoting the state and its time derivatives, $\mathbf M(\mathbf q) \in \mathbb{R}^{n \times n}$ the mass/inertia matrix, $\mathbf C(\mathbf q, \dot{\mathbf q}) \in \mathbb{R}^{n \times n}$ the Coriolis matrix, $\mathbf G (\mathbf q) \in \mathbb{R}^n$ the gravity term, $\mathbf F (\dot{\mathbf{q}}) \in \mathbb{R}^n$ the damping/friction term, $\mathbf d \in \mathbb{R}^n$ the external disturbance, and $\mathbf u \in \mathbb{R}^n$ the control input.

The Euler-Lagrange dynamics (\ref{EL}) can describe several real-world mechanical systems, such as robotic manipulators and aerial/ground/marine vehicles \cite{9001188,2020Ona, 9715173}. 
For all such systems, the following standard properties hold: 

\vspace{0.15cm}
\noindent \textbf{Property 1:} There exist $c, g, f, d \in \mathbb{R}^{+}$ such that $||\mathbf C(\mathbf q,\dot{\mathbf q})|| \leq c ||\dot{\mathbf q}||$, $||\mathbf G \mathbf{(q)}|| \leq  g$, $||\mathbf F(\dot{\mathbf q})|| \leq {\color{black}f}||\dot{\mathbf q}||$ and $||\mathbf d(t)|| \leq d$.

\vspace{0.15cm}
\noindent \textbf{Property 2:} The matrix $\mathbf M \mathbf{(q)}$ is symmetric and uniformly positive, that is, $\exists$ $m_1, m_2 \in \mathbb{R}^{+}$ such that
\begin{equation}\label{prop 2}
0 < m_1 \mathbb I_n  \leq \mathbf M \mathbf{(q)} \leq m_2 \mathbb I_n.
\end{equation}

\noindent \textbf{Property 3:} The matrix $\dot {\mathbf M} (\mathbf q) - 2\mathbf C(\mathbf{q}, \dot {\mathbf{q}})$ is skew symmetric i.e. $\mathbf x^{T}(\dot {\mathbf M} (\mathbf q) - 2\mathbf C(\mathbf{q}, \dot {\mathbf{q}})) \mathbf{x}=0$ for any non-zero vector $\mathbf{x}$.

\begin{remark}[Unstructured uncertainty]\label{remark1}
Properties 1-3 do not rely on the {\color{black}specific} structure of the system. 
 {\color{black}The structure of} $\mathbf{M,C,F,G,d}$ in (\ref{EL}) and  {\color{black}the knowledge of}  $m_1, m_2, c, g,f,d$ in their corresponding bounds {\color{black}will be unknown in this work.}
\end{remark}

%
%
\noindent Let us recall {\color{black}and extend} a Lyapunov characterization of practical finite-time stability, to be used for stability analysis:
\begin{lemmaNoParens}[\cite{7959105}] \label{lemmafinitetime}
Consider general nonlinear dynamics
\begin{align} \label{nonlinearmode}
    \dot x = {\color{black}\hbar}(t,x), \quad x(0) = x_0,
\end{align}
with 
${\color{black}\hbar}:\mathbb{R}_{+} \times \mathbb{R}^n \mapsto \mathbb{R}^n$ a continuous function satisfying $g(t,0)=0, \forall t$. Suppose there exists a positive definite and continuous radially unbounded  function $V(x)\hspace{-0.1cm}: \mathbb{R}^n \mapsto \mathbb{R}$ such that $\dot{V}(x) \leq -\eta V^{\color{black}\frac{1+v}{2}}(x)+{\color{black}\omega}$, where $\eta>0$, {\color{black}$0<v\le 1$}, ${\color{black}\omega} < \infty$. Then, the origin of (\ref{nonlinearmode}) is practical finite-time stable and the state $x$ converges in finite time to the region 
{\color{black}
\begin{align*}
    &\hspace{1.0cm}\Psi = \{x| \ V^{\color{black}\frac{1+v}{2}}(x) \leq \frac{{\color{black}\omega}}{(1-\iota)\eta} \},  \quad 0<\iota<1     
\end{align*}
where an estimate for the finite time $T_s$ is \vspace{-0.1cm} 
\begin{align*}
    & \begin{matrix}  T_s \hspace{-0.05cm} =\hspace{-0.05cm} \frac{2}{(1\hspace{-0.03cm}-\hspace{-0.03cm}{\color{black}v})\iota \eta}\hspace{-0.1cm}\left[V^{\hspace{-0.03cm}{\color{black}\frac{1-v}{2}}}(x(0))\hspace{-0.09cm}-\hspace{-0.09cm} \left (\frac{{\color{black}\omega}}{(1\hspace{-0.03cm}-\hspace{-0.03cm}\iota)\eta} \right )^{\color{black}\frac{1-v}{1+v}}\right],  & \hspace{-0.09cm}\text{if} \  0\hspace{-0.05cm}<\hspace{-0.05cm} {\color{black}v} \hspace{-0.05cm}<\hspace{-0.05cm}1,    \end{matrix}
\end{align*}
} {\color{black}and, by using L'Hôpital's rule for $\lim_{v \to 1} T_s$,}
\begin{align*}
     \begin{matrix}  \hspace{-0.1cm}T_s \hspace{-0.05cm}= \frac{1}{\iota \eta}[\ln{V(x(0)){\color{black}\omega}}-\ln{(1-\iota)\eta}], & \hspace{-0.15cm} \qquad \ \text{if} \  {\color{black}v}=1.    \quad   \end{matrix} 
\end{align*}
\end{lemmaNoParens}

\noindent Denote with $\mathbf q_d, \dot{\mathbf{q}}_d, \ddot{\mathbf{q}}_d \in \mathbb{R}^n$ the desired state and its derivatives, which satisfy $\| \dot{\mathbf{q}}_d(t) \| \leq q_m$,   $\|\ddot{\mathbf{q}}_d(t)\| \leq q_{mm}$, $\forall t$. The following control question arises:\vspace{0.15cm}

\noindent\textbf{Problem:} Without structural knowledge of the terms in   (\ref{EL}) and of their bounds as in Remark \ref{remark1},  find a control $\mathbf u(\cdot)$ such that the origin of the EL closed loop is practical finite-time stable with finite-time convergence of the tracking error $\mathbf e = \mathbf q - \mathbf q_d$.

\section{Sliding surface and adaptive control design}
The design is organized as follows: first, a novel sliding surface is proposed; 
then, the uncertainties affecting the error dynamics are analyzed (Sect. \ref{Sec: Uncertainty}), leading to an adaptive law to compensate them (Sect. \ref{Sec: AdativeLaw}). 
The proposed sliding surface is 
\begin{equation} \label{surface}
    \mathbf{s} = \dot{\mathbf e} + \boldsymbol \lambda_1 \mathbf{e} + \boldsymbol \lambda_2\int_0^t \mathbf e(\tau)d\tau + \boldsymbol \lambda_3 \Delta(\mathbf{e}),
\end{equation}
where $\boldsymbol \lambda_k, k = 1, 2, 3$ are positive definite diagonal matrices with {\color{black}entries} $\lambda_{ki}$, $i = 1,2, \cdots, n$, and $\Delta(\mathbf{e})$ is a vector with {\color{black}entries} $\Delta_i (e_i)$  defined as
\begin{align} \label{Delta_i}
    \Delta_i (e_i) = \left \{ \begin{matrix} \sig^{\gamma}(e_i),    & \text{if} \  |e_i| > \varepsilon,   \\ \alpha_1 e_i + \alpha_2 \sign(e_i) e^2_i,  & \text{if} \  |e_i| \le \varepsilon. \end{matrix} \right.
\end{align}
where $0<\gamma<1$, $\varepsilon$ is a small positive constant, and
\begin{align*}
    \alpha_1=(2-\gamma) \varepsilon^{\gamma-1}, \qquad
    \alpha_2=(\gamma-1) \varepsilon^{\gamma-2}.
\end{align*}
The time derivative of $\mathbf s$ 
is well defined in the entire state space 
\begin{equation} \label{s_dot}
    \dot{\mathbf{s}} = \ddot{\mathbf{e}} + \boldsymbol \lambda_1 \dot{\mathbf{e}} + \boldsymbol \lambda_2 \mathbf{e} + \boldsymbol \lambda_3 \dot \Delta(\mathbf e),
\end{equation}
with the {\color{black}entries}  of $\dot \Delta(\mathbf{e})$ being
\begin{align} \label{Delta_i_dot}
    \dot \Delta_i (e_i) = \left \{ \begin{matrix} \gamma |e_i|^{\gamma-1} \dot{e}_i ,   & \text{if} \  |e_i| > \varepsilon   \\ (\alpha_1 + 2\alpha_2 |e_i|) \dot e_i,  &  \hspace{0.08cm} \text{if} \  |e_i| \le \varepsilon. \end{matrix} \right.
\end{align}
Note that the parameters $\alpha_1$, $\alpha_2$ ensure continuity of $\mathbf s$ and $\dot{\mathbf s}$ at the point $|e_i|=\varepsilon$. 
The following lemma provides appropriate bounds on $\Delta(\mathbf{e})$ and $\dot \Delta(\mathbf{e})$, useful for stability analysis:
\vspace{-0.2cm}
\begin{lemmaNoParens} \label{inequa_Delta}
       For $\Delta(\mathbf e)$, $\dot \Delta(\mathbf e)$ in  (\ref{Delta_i}) and (\ref{Delta_i_dot}), the following holds
    \begin{align} \label{bound_Delta}
        \|\Delta(\mathbf{e})\| \le \|\sig^\gamma(\mathbf{e})\|, \quad
        \|\dot \Delta(\mathbf{e})\| \le \alpha_1\|\dot{\mathbf e}\|.
    \end{align}
\end{lemmaNoParens}
\begin{proof} To prove the first bound in (\ref{bound_Delta}), substitute $\alpha_1$, $\alpha_2$ into (\ref{Delta_i}) and define the function ${\color{black}h}(|e_i|, \varepsilon)= (2-\gamma)\varepsilon^{\gamma-1}|e_i|+(\gamma-1)\varepsilon^{\gamma-2}|e_i|^2,  \varepsilon \ge |e_i|$. We have 
\begin{align*}
    \frac{\partial {\color{black}h}(|e_i|,\varepsilon)}{\partial \varepsilon} = (\gamma-1)(\gamma-2)|e_i|\varepsilon^{\gamma-3}(|e_i|-\varepsilon),
\end{align*}
which implies that the gradient of ${\color{black}h}(|e_i|,\varepsilon)$ with respect to $\varepsilon$ is negative if $\varepsilon \ge |e_i|$. So $f(|e_i|,\varepsilon)$ is monotonically decreasing with respect to $\varepsilon$, which implies 
\begin{equation} \label{Delta_i_bound}
    f(|e_i|,\varepsilon) \le f(|e_i|,|e_i|) = |e_i|^\gamma.
\end{equation}
Substituting (\ref{Delta_i_bound}) into (\ref{Delta_i}) gives
$|\Delta_i(e_i)| \le |e_i|^\gamma$, from which we conclude that $\|\Delta(\mathbf{e})\| \le \|\sig^\gamma(\mathbf{e})\|$. 

To prove the second bound in (\ref{bound_Delta}), define the function
\begin{equation*}
    {\color{black}\underline{h}}(|e_i|)=\left \{ \begin{matrix} \gamma |e_i|^{\gamma-1} ,   & \text{if} \  |e_i| > \varepsilon,   \\ \alpha_1 + 2\alpha_2 |e_i|,  & \hspace{0.03cm} \text{if} \  |e_i| \le \varepsilon. \end{matrix} \right.
\end{equation*}
The function ${\color{black}\underline{h}}(|e_i|)$ is monotonically decreasing with respect to $|e_i|$ because of $0<\gamma<1$ and $\alpha_2 < 0$, which implies 
\begin{equation} \label{Delta_i_dot_bound}
    {\color{black}\underline{h}}(|e_i|) \le g(0) = \alpha_1.
\end{equation}
Substituting (\ref{Delta_i_dot_bound}) into (\ref{Delta_i_dot}) gives $|\dot \Delta_i(e_i)| \le \alpha_1|\dot e_i|$, from which we conclude that $\| \dot \Delta(\mathbf{e})\| \le \alpha_1\| \dot {\mathbf{e}}\|$. 
\end{proof}

\subsection{Uncertainty structure} \label{Sec: Uncertainty}
Based on the proposed sliding surface (\ref{surface}), a suitable upper bound structure for the uncertainty of (\ref{EL}) will be derived. Multiplying (\ref{s_dot}) by $\mathbf{M}$ and using (\ref{EL}) we obtain
\begin{align} \label{Ms_dot}
    \mathbf{M}\dot{\mathbf{s}} &=  \mathbf{M} \boldsymbol \lambda_1\dot{\mathbf{e}} + \mathbf{M} \boldsymbol \lambda_2\mathbf{e} + \mathbf{M} \boldsymbol \lambda_3\dot \Delta(\mathbf e) +  \mathbf{M}(\ddot{\mathbf{q}}-\ddot{\mathbf{q}}_d) \nonumber \\
    &= \mathbf{u} - \mathbf{Cs} + \boldsymbol \varphi,
\end{align}\vspace{-0.1cm}
where
\begin{align}
    \boldsymbol \varphi &\triangleq  -(\mathbf C \dot{\mathbf{q}} + \mathbf G + \mathbf F + \mathbf d \nonumber \\
    & \quad + \mathbf M \ddot{\mathbf q}_d - \mathbf{M}\boldsymbol\lambda_1\dot{\mathbf{e}} -  \mathbf{M}\boldsymbol\lambda_2\mathbf{e} - \mathbf{M}\boldsymbol \lambda_3\dot \Delta(\mathbf e) - \mathbf C \mathbf{s}),
\end{align}
represents an aggregate state-dependent uncertainty. 
By combining (\ref{surface}) and Properties 1-2, we obtain the following (state-dependent) bound for such uncertainty  
\begin{align} \label{bound}
    \hspace{-0.25cm}\|\boldsymbol{\varphi}\| & \leq c\|\dot{\mathbf{q}}\|^2 + g + {\color{black}f}\|\dot{\mathbf{q}}\| + d \nonumber \\  & \quad+ 
     m_2(\|\ddot{\mathbf{q}}_d\| + \|\boldsymbol\lambda_1\|\|\dot{\mathbf e}\| + \|\boldsymbol\lambda_2\|\|\mathbf{e}\| + \|\boldsymbol\lambda_3\|\|\dot \Delta(\mathbf{e})\|)  \nonumber \\  &  \quad+ 
    c\|\dot{\mathbf{q}}\|(\|\dot{\mathbf e}\| + \|\boldsymbol\lambda_1\|\|\mathbf{e}\| + \|\boldsymbol\lambda_2\|\|\int_0^t{\mathbf{e}(\tau)d\tau}\|  \nonumber \\ & \quad  + \|\boldsymbol\lambda_3\|\|\Delta(\mathbf{e})\|).
\end{align}
Define
$$\boldsymbol{\xi}(t) = \begin{bmatrix}(\mathbf{e} + \sig^\gamma(\mathbf e))^T(t) & \dot{\mathbf e}^T(t) &(\int_0^t{\mathbf{e}(\tau)d\tau})^T \end{bmatrix}^T.$$
Clearly, the following inequalities hold 
\begin{align} \label{inequalities}
\begin{split}
\|\boldsymbol \xi(t)\| &\geq \|\mathbf{e}(t)\|, \quad \|\boldsymbol \xi(t)\| \geq \|\sig^\gamma(\mathbf e)(t)\|,  \\
\|\boldsymbol \xi(t)\| &\geq \|\dot{\mathbf{e}}(t)\|, \quad \|\boldsymbol \xi(t)\| \geq \|\int_0^t{\mathbf{e}(\tau)d\tau}\|. 
\end{split}
\end{align}
According to Lemma \ref{inequa_Delta} and (\ref{inequalities}), we have
\begin{align}\label{linearbound}
    \|\boldsymbol \varphi\| &\leq \theta_0^* + \theta_1^*\|\boldsymbol \xi(t)\|+ \theta_2^*\|\boldsymbol \xi(t)\|^2,
\end{align}
where $\theta_0^*$, $\theta_1^* $, $\theta_2^*$ are positive {\color{black}unknown} constants defined as
\begin{align*}
\theta_0^* &\triangleq cq_m^2 + g+{\color{black}f}q_m + d + m_2q_{mm},\\
    \theta_1^* &\hspace{-0.1cm}\triangleq \hspace{-0.1cm} c q_m(\hspace{-0.05cm}3 \hspace{-0.1cm}+\hspace{-0.1cm} \|\boldsymbol\lambda_1\| \hspace{-0.1cm}+\hspace{-0.1cm} \|\boldsymbol\lambda_2\| \hspace{-0.1cm}+\hspace{-0.1cm} \|\boldsymbol\lambda_3\|\hspace{-0.05cm}) \hspace{-0.1cm}+\hspace{-0.12cm} {\color{black}f} \hspace{-0.12cm}+\hspace{-0.1cm}  m_2(\hspace{-0.03cm}\|\boldsymbol\lambda_1\| \hspace{-0.1cm}+\hspace{-0.1cm} \|\boldsymbol\lambda_2\| \hspace{-0.1cm}+\hspace{-0.1cm} \|\boldsymbol\lambda_3\|\alpha_1\hspace{-0.05cm}), \\ 
    \theta_2^* & \triangleq c(2+\|\boldsymbol\lambda_1\| +\|\boldsymbol\lambda_2\| +\|\boldsymbol\lambda_3\|).
\end{align*}
Note {\color{black}that the form \eqref{linearbound} has been derived independently on the specific system structure.} 

\subsection{Adaptive control law} \label{Sec: AdativeLaw}
Aiming at compensating the uncertainty (\ref{linearbound}), the following control input is proposed
\begin{align} \label{controller}
    \mathbf u(t) = -\boldsymbol \Lambda \mathbf{s}(t) - \rho(t) {\color{black}\sign}({\mathbf{s}(t)}) -\sigma\sig^{v}(\mathbf{s}(t)),
\end{align}
with $\boldsymbol \Lambda$ a positive definite matrix, $\mu \hspace{-0.025cm}>\hspace{-0.025cm}0, \sigma \hspace{-0.025cm}>\hspace{-0.025cm}0$, ${\color{black}0\hspace{-0.025cm}<\hspace{-0.025cm}v\hspace{-0.025cm}\le\hspace{-0.025cm} 1}$,
and
\begin{align}\label{rho}
\rho(t) = \hat{\theta}_0(t) + \hat{\theta}_1(t)||\boldsymbol \xi(t)|| + \hat{\theta}_2(t)||\boldsymbol \xi(t)||^2. 
\end{align}
The gains $\hat{\theta}_i,i=0, 1, 2$ can be interpreted as estimates of $\theta^*_i$ in (\ref{linearbound}), updated by the following adaptive laws
\begin{equation} \label{ThetaadAptiveLaw}
    \dot{\hat{\theta}}_i(t) = \|\mathbf{s}(t)\|\|\boldsymbol \xi(t)\|^i-\beta_i \hat{\theta}_i^v(t), \quad \hat{\theta}_i(0) > 0
\end{equation}
where $\beta_i, i=0,1,2$ are positive constants. 
We now analyze the stability properties of the proposed method.


\begin{lemmaNoParens} \label{ThetaInitial}
    Consider the adaptive laws in (\ref{ThetaadAptiveLaw}) with initial condition $\hat{\theta}_i(0)>0$. Then, $\hat{\theta}_i(t)\ge0, \forall t \ge 0$. 
\end{lemmaNoParens}
\begin{proof} We use a proof by contradiction. Assume that $\exists \hat{\theta}_i(t)<0$ under initial condition $\hat{\theta}_i(0)>0$. Then, there must exist $0<\tau<t$ such that $\hat{\theta}_i(\tau)=0$ and $\dot{\hat{\theta}}_i(\tau)<0$ because of continuity of $\hat{\theta}_i(t)$. However, substituting $\hat{\theta}_i(\tau)=0$ into (\ref{ThetaadAptiveLaw}) gives $\dot{\hat{\theta}}_i(\tau) = \|\mathbf{s}(\tau)\|\|\boldsymbol \xi(\tau)\|^i \ge 0$, which does not satisfy the previously assumed negative time derivative. This verifies that $\hat{\theta}_i(t)\ge0, \forall t \ge 0$ hold under initial condition $\hat{\theta}_i(0) > 0$.
\end{proof}
\begin{theorem}
     Under Properties 1-3, the origin of the closed-loop Euler-Lagrange  dynamics (\ref{EL}) with control law (\ref{controller})-(\ref{rho}) and adaptive law (\ref{ThetaadAptiveLaw}), are practical finite-time stable. Let $\bar{s}$ be the bound for $\mathbf{s}$ in \eqref{surface} after finite time: then, the tracking error $\mathbf{e}$ converges in finite time to the region 
     \begin{align*}
\begin{split}
             \Phi &= \max \{  \sqrt{n}\varepsilon, \frac{\sqrt{n}\bar{s}}{\min_i\{\lambda_{1i}\}} ,  \sqrt{n}\left (\frac{\bar{s}}{\min_i\{\lambda_{3i}\}} \right )^{\frac{1}{\gamma}} \} .             
         \end{split}
     \end{align*}
\end{theorem}
\begin{proof} Select the Lyapunov function candidate:
\begin{equation}
    V(\mathbf{s},\tilde{\theta}) = \frac{1}{2}\mathbf s^T \mathbf{M} \mathbf{s} + \frac{1}{2} \sum_{i=0}^2 \tilde{\theta}_i^2 \label{Lyapunov}
\end{equation}
with $\tilde{\theta}_i = \theta_i^*-\hat{\theta}_i$.  The time derivative of $V$ along the dynamics (\ref{Ms_dot}), (\ref{controller}) is
\begin{align*}
    \dot V &= \mathbf{s}^T(\mathbf{u}-\mathbf{Cs}+\boldsymbol \varphi) + \frac{1}{2}\mathbf{s}^T \dot{\mathbf{M}} \mathbf{s} - \sum_{i=0}^2\tilde{\theta}_i\dot{\hat{\theta}}_i  \\
    &=\hspace{-0.05cm} \mathbf{s}^T \hspace{-0.08cm}(\hspace{-0.05cm}-\hspace{-0.08cm}\boldsymbol \Lambda \mathbf{s}\hspace{-0.1cm} - \hspace{-0.1cm}\rho {\color{black}\sign}({\mathbf{s}}) \hspace{-0.08cm} -\hspace{-0.08cm}\sigma\sig^{v}\hspace{-0.06cm}(\mathbf{s}) \hspace{-0.1cm}+\hspace{-0.1cm}\boldsymbol \varphi\hspace{-0.05cm}) \hspace{-0.11cm} + \hspace{-0.11cm}\frac{1}{2}\mathbf s^T\hspace{-0.07cm}(\dot{\mathbf{M}}\hspace{-0.1cm} - \hspace{-0.1cm} 2\mathbf{C}) \mathbf s \hspace{-0.1cm}-\hspace{-0.2cm} \sum_{i=0}^2\tilde{\theta}_i\dot{\hat{\theta}}_i.
\end{align*}
Using Property 3, we obtain
\begin{align*}
    \dot V &= -\mathbf{s}^T\boldsymbol \Lambda \mathbf{s} \hspace{-0.05cm} -\hspace{-0.05cm}\sigma\mathbf{s}^T\sig^{v}(\mathbf{s})\hspace{-0.05cm} +\hspace{-0.05cm}\mathbf{s}^T \boldsymbol \varphi - \rho \mathbf{s}^T {\color{black}\sign}({\mathbf{s}}) \hspace{-0.05cm}-\hspace{-0.05cm} \sum_{i=0}^2\tilde{\theta}_i\dot{\hat{\theta}}_i.
\end{align*}
Since $\rho > 0$ from Lemma \ref{ThetaInitial}, substituting (\ref{linearbound}), (\ref{rho}) and (\ref{ThetaadAptiveLaw}) in the Lyapunov time derivative gives
\begin{align}
    \dot V &\le -\mathbf{s}^T\boldsymbol \Lambda \mathbf{s}  -\sigma\mathbf{s}^T\sig^{v}(\mathbf{s}) +\sum_{i=0}^2{\tilde{\theta}_i(\|\boldsymbol \xi\|^i\|\mathbf s\| - \dot{\hat{\theta}}_i)} \nonumber \\
    &\le -\mathbf{s}^T\boldsymbol \Lambda \mathbf{s}  -\sigma\sum_{i=1}^n{(|s_i|^2)^{\frac{1+v}{2}}} + \sum_{i=0}^2{\beta_i\tilde{\theta}_i \hat{\theta}_i^v}. \label{simplifyVdot}
\end{align}
{\color{black}In line with Lemma 1, we cover the cases $0<v<1$ and $v=1$, to make the design more flexible.}

{\color{black} \emph{Design (a)}: $0<v<1$.} 
Using Lemma \ref{ThetaInitial} and the definitions of $\theta^*_i$, $\tilde{\theta}_i$, we have $\theta^*_i>0$, $\hat{\theta}_i \ge 0$, and $\tilde{\theta}_i \le \theta^*_i$. Then, let us analyze the term $\tilde{\theta}_i \hat{\theta}^v_i$ by making use of {\color{black}\cite[Lemmas 4,5]{2018Adaptive}}
\begin{align} \label{inequatheta}
    \hspace{-0.3cm}\tilde{\theta}_i \hat{\theta}^v_i &\hspace{-0.12cm}=\hspace{-0.08cm} \hat{\theta}^v_i (\theta^*_i-\hat{\theta}_i) \nonumber \\ &\hspace{-0.12cm}\le\hspace{-0.1cm} \frac{1}{1\hspace{-0.08cm}+\hspace{-0.1cm}v} ({\theta_i^{*^{1+v}}}\hspace{-0.3cm}-\hat{\theta}_i^{1+v}) \leq \frac{1}{1\hspace{-0.08cm}+\hspace{-0.1cm}v} ({\theta_i^{*^{1+v}}}\hspace{-0.3cm}-(\theta_i^*-|\tilde{\theta}_i|)^{1+v}) \nonumber \\
    &\hspace{-0.12cm}\le \hspace{-0.1cm} \frac{1}{1\hspace{-0.1cm}+\hspace{-0.1cm}v}({\theta_i^{*^{1+v}}}\hspace{-0.35cm}-\hspace{-0.05cm}(|\tilde{\theta}_i|^{1+v} \hspace{-0.3cm}-{\theta_i^{*^{1+v}}})) \hspace{-0.08cm}=\hspace{-0.08cm} \frac{1}{1\hspace{-0.1cm}+\hspace{-0.1cm}v}(2 \theta_i^{*^{1+v}}\hspace{-0.3cm}-\hspace{-0.08cm}|\tilde{\theta}_i|^{1+v}).
\end{align}

\noindent Substituting (\ref{inequatheta}) into (\ref{simplifyVdot}) gives 
\begin{align*}
        \dot V &\le -\lambda_{min}(\boldsymbol \Lambda)\|\mathbf{s}\|^2 - \sigma\sum_{i=1}^n{(|s_i|^2)^{\frac{1+v}{2}}} \\ & \quad- \frac{1}{1+v}\sum_{i=0}^2\beta_i{|\tilde{\theta}_i|^{1+v}} + \frac{2}{1+v}\sum_{i=0}^2{\beta_i\theta_i^{*^{1+v}}} \\
        & \le -\hspace{-0.02cm} \sigma\hspace{-0.1cm} \sum_{i=1}^n{\hspace{-0.05cm}(|s_i|^2)^{\frac{1+v}{2}}} \hspace{-0.14cm}-\hspace{-0.1cm}  \frac{1}{1\hspace{-0.1cm}+\hspace{-0.1cm}v}\hspace{-0.08cm} \sum_{i=0}^2\hspace{-0.05cm}\beta_i{(|\tilde{\theta}_i|^2)^{\frac{1+v}{2}}} \hspace{-0.1cm}+\hspace{-0.1cm} \frac{2}{1\hspace{-0.08cm}+\hspace{-0.08cm}v} \hspace{-0.08cm} \sum_{i=0}^2{\hspace{-0.05cm}\beta_i\theta_i^{*^{1+v}}}.
\end{align*}
Then, using {\color{black}\cite[Lemma 1]{ZUO2015305}} we have 
\begin{align} \label{Vdot}
    \dot V \le -\eta_1 \left (\frac{m_2}{2}\sum_{i=1}^n|s_i|^2 + \frac{1}{2}\sum_{i=0}^2|\tilde{\theta}_i|^2 \right)^{\frac{1+v}{2}} \hspace{-0.3cm} + \omega_1,
\end{align}
where $\eta_1\hspace{-0.1cm}=\hspace{-0.1cm}\min_i \{\frac{2\sigma}{m_2},\frac{2\beta_i}{1+v} \}\hspace{-0.05cm}>\hspace{-0.05cm}0$ and  $\omega_1\hspace{-0.1cm}=\hspace{-0.1cm} \frac{2}{1+v}\sum_{i=0}^2{\beta_i\theta_i^{*^{1+v}}}\hspace{-0.2cm}>0$. The Lyapunov function (\ref{Lyapunov}) can be upper bounded as
\begin{equation} \label{inequaLya}
    V \le \frac{m_2}{2}\|\mathbf{s}\|^2 + \frac{1}{2}\sum_{i=0}^2 |\tilde{\theta}_i|^2,
\end{equation}
and substituting (\ref{inequaLya}) into (\ref{Vdot}) gives
\begin{equation}
    \dot V \le -\eta_1 V^{\frac{1+v}{2}} + \omega_1.
\end{equation}
We conclude from  Lemma \ref{lemmafinitetime} that the states $(\mathbf{s},\tilde{\theta})$ are driven to the region $\Psi$ in finite-time $T_s$, where
\begin{align*}
    \Psi &= \{\mathbf{s, \tilde{\theta}}| \ V^{\frac{1+v}{2}}(\mathbf{s, \tilde{\theta}}) \leq \frac{\omega_1}{(1-\iota)\eta_1} \}, \quad 0<\iota<1, \\
    T_s &= \frac{2}{(1\hspace{-0.02cm}-\hspace{-0.02cm}v)\iota \eta_1} \left [V^{\frac{1-v}{2}}(\mathbf{s}(0), \tilde{\theta}(0)) -\left (\frac{\omega_1}{(1\hspace{-0.02cm}-\hspace{-0.02cm}\iota)\eta_1} \right )^{\frac{1-v}{1+v}} \right].
\end{align*}

{\color{black} \emph{Design (b)}: $v=1$. Note that (\ref{simplifyVdot}) reduces to
\begin{align*}
    \dot V &\le -\mathbf{s}^T\boldsymbol \Lambda \mathbf{s} + \sum_{i=0}^2{\beta_i\tilde{\theta}_i \hat{\theta}_i}.
\end{align*}
With analogous steps as design (a), we have that the states $(\mathbf{s},\tilde{\theta})$ are driven to the region $\Psi$ in finite-time $T_s$, where
\begin{align*}
    \Psi &= \{\mathbf{s, \tilde{\theta}}| \ V^{\frac{1+v}{2}}(\mathbf{s, \tilde{\theta}}) \leq \frac{\omega_2}{(1-\iota)\eta_2} \}, \\
    T_s &= \frac{1}{\iota \eta_2}[\ln{V(\mathbf{s}(0), \tilde{\theta}(0))\omega_2}-\ln{(1-\iota)\eta_2}],
\end{align*}
where $\eta_2\hspace{-0.1cm}=\hspace{-0.1cm}\frac{\min_i\{\lambda_{min}(\boldsymbol \Lambda), \beta_i/2\}}{\max_i\{{m_2}/{2},1/2\}} \hspace{-0.1cm}>\hspace{-0.1cm} 0$ and  $\omega_2\hspace{-0.08cm}=\hspace{-0.08cm} \frac{1}{2}\sum_{i=0}^2\beta_i \theta_i^{*2} > 0$.} 

To find a bound for the tracking error, consider the Lyapunov function $V_e = \frac{1}{2}e_i^2 + \frac{1}{2}\lambda_{2i}(\int_0^t e_i(\tau)d\tau)^2$, giving 
\begin{align} \label{V_edot}
    \dot{V}_e = e_i \dot{e}_i + \lambda_{2i} e_i \int_0^t e_i(\tau)d\tau.
\end{align}
Then, we proceed along two cases as follows. 

\emph{Case I}: $|e_i| > \varepsilon$. According to (\ref{surface}) and (\ref{Delta_i}) we have
\begin{equation} \label{surface_e}
    \dot{e}_i + \lambda_{1i} e_i + \lambda_{2i} \int_0^t{e_i(\tau)d\tau} +\lambda_{3i} \sig^\gamma(e_i) = s_i,
\end{equation}
which can be rewritten in three different forms as
\begin{align}
    &\dot{e}_i + (\lambda_{1i}\hspace{-0.05cm}-\hspace{-0.05cm}\frac{s_i}{e_i})e_i + \lambda_{2i} \int_0^t{e_i(\tau)d\tau} + \lambda_{3i} \sig^\gamma(e_i) = 0, \label{form1} \\
    &\dot{e}_i \hspace{-0.1cm} + \hspace{-0.1cm} \lambda_{1i} e_i \hspace{-0.1cm} + \hspace{-0.1cm} (\lambda_{2i}\hspace{-0.08cm}-\hspace{-0.08cm}\frac{s_i}{\int_0^t\hspace{-0.08cm}{e_i(\tau)d\tau}})\hspace{-0.15cm} \int_0^t\hspace{-0.19cm}{e_i(\tau)d\tau} \hspace{-0.08cm} + \hspace{-0.1cm} \lambda_{3i} \sig^\gamma(e_i) \hspace{-0.08cm}=\hspace{-0.08cm} 0, \label{form2} \\
    &\dot{e}_i \hspace{-0.08cm}+\hspace{-0.08cm} \lambda_{1i} e_i \hspace{-0.08cm} +\hspace{-0.08cm}\lambda_{2i}\hspace{-0.08cm} \int_0^t\hspace{-0.08cm}{e_i(\tau)d\tau}\hspace{-0.08cm}+ \hspace{-0.08cm} (\lambda_{3i}\hspace{-0.08cm}-\hspace{-0.08cm}\frac{s_i}{\sig^\gamma(e_i)})\sig^\gamma(e_i) \hspace{-0.07cm} = \hspace{-0.07cm} 0. \label{form3}
\end{align}
Note that substituting the condition $s_i = 0$ into (\ref{surface}) gives
\begin{align} \label{e_idot}
    \dot{e}_i = -(\lambda_{1i} e_i + \lambda_{2i} \int_0^t e_i(\tau)d\tau + \lambda_{3i} \Delta_i(e_i) ).
\end{align}
Then, substituting (\ref{e_idot}) into (\ref{V_edot}) gives
\begin{align*}
     \dot{V}_e \hspace{-0.05cm} &=\hspace{-0.05cm} -e_i(\lambda_{1i} e_i\hspace{-0.05cm} +\hspace{-0.05cm} \lambda_{3i} \Delta_i(e_i)) \hspace{-0.05cm}=\hspace{-0.05cm} -\lambda_{1i} e_i^2\hspace{-0.05cm} - \hspace{-0.05cm}\lambda_{3i} |e_i|^{\gamma+1} \hspace{-0.05cm}\le \hspace{-0.05cm}0.
\end{align*}
In the following, let us analyze the tracking errors assuming (\ref{e_idot}) holds. The fact that $V_e > 0$ and $\dot V_e \le 0$ 
 implies $\lim_{t \to \infty}V_e(t) = V_e(\infty) < V_e(0)$, that is, $e_i$ and $\int_0^t e_i(\tau)d\tau$ are bounded. Meanwhile, $\dot e_i$ is also bounded according to (\ref{e_idot}). 
Furthermore, by integrating $\dot{V}_e$, we have
\begin{align*}
    \lim_{t \to \infty}\int_0^t{\lambda_{1i} e_i^2(\tau) + \lambda_{3i} |e_i|^{\gamma+1}(\tau)d\tau} = V_e(0) - V_e(\infty),
\end{align*}
which implies that $e_i$ has bounded 2-norm. 
Thus, assuming (\ref{e_idot}) holds, we get $e_i\to 0$ from Barbalat's lemma. An ultimate bound for $\int_0^t e_i(\tau)d\tau$ is obtained noting that, for zero initial conditions of the integrator, $V_e(\infty) \hspace{-0.1cm}<\hspace{-0.1cm} V_e(0) \hspace{-0.1cm}=\hspace{-0.1cm} \frac{1}{2}e_i^2(0)$, and $V_e(\infty) = \frac{1}{2}e_i^2(\infty) + \lim_{t \to \infty}\frac{1}{2}\lambda_{2i}(\int_0^t\hspace{-0.05cm}{ e_i^2(\tau)d\tau})^2 \hspace{-0.1cm}<\hspace{-0.1cm} \frac{1}{2}e_i^2(0)$. This gives $\frac{|e_i(0)|}{\sqrt{\lambda_{2i}}}$ as an ultimate bound for $|\int_0^t{e_i(\tau)d\tau}|$. 

With this in mind, we notice that (\ref{form1}), (\ref{form2}) and (\ref{form3}) are in the form of (\ref{e_idot}) whenever 
\begin{align*}
    \lambda_{1i}\hspace{-0.05cm}-\hspace{-0.05cm}\frac{s_i}{e_i} \hspace{-0.05cm}>\hspace{-0.05cm} 0, \    \ \lambda_{2i}\hspace{-0.05cm}-\hspace{-0.05cm}\frac{s_i}{\int_0^t{e_i(\tau)d\tau}}\hspace{-0.05cm}>\hspace{-0.05cm} 0, \   \ 
    \lambda_{3i}\hspace{-0.05cm}-\hspace{-0.05cm}\frac{s_i}{\sig^\gamma(e_i)} \hspace{-0.05cm}>\hspace{-0.05cm} 0.
\end{align*}
This implies that  $e_i$ will converge to\vspace{-0.2cm} 
\begin{align*}
    &|e_i| \le \frac{\bar{s}}{\min_i\{\lambda_{1i}\}} \ \ \text{or} \ \
    |e_i| \le \left (\frac{\bar{s}}{\min_i\{\lambda_{3i}\}} \right )^{\frac{1}{\gamma}},  
\end{align*}
whereas $\int_0^t{e_i(\tau)d\tau}$ will converge to 
\begin{align*}
    &|\hspace{-0.075cm}\int_0^t{e_i(\tau)d\tau}| \hspace{-0.025cm}\le\hspace{-0.025cm} \frac{\bar{s}}{\min_i\{\lambda_{2i}\}} \ \ \hspace{-0.025cm}\text{or}\hspace{-0.025cm} \ \
    |\hspace{-0.075cm}\int_0^t{e_i(\tau)d\tau}| \hspace{-0.025cm}\le\hspace{-0.025cm} \max_i\{\frac{|e_i(0)|}{\sqrt{\lambda_{2i}}}\}.
\end{align*}
\emph{Case II}: $|e_i| \le \varepsilon$. Substituting (\ref{Delta_i}) into (\ref{surface}) gives
\begin{equation} \label{surface_e2}
    \dot{e}_i \hspace{-0.035cm}+\hspace{-0.035cm} \lambda_{1i} e_i \hspace{-0.035cm}+\hspace{-0.035cm} \lambda_{2i}\hspace{-0.1cm}\int_0^t\hspace{-0.15cm} e_i(\tau)d\tau \hspace{-0.035cm}+\hspace{-0.035cm} \lambda_{3i} (\alpha_1 e_i \hspace{-0.035cm}+\hspace{-0.035cm} \alpha_2 \sign(e_i) e^2_i) \hspace{-0.05cm}=\hspace{-0.05cm} s_i.
\end{equation}
Along analogous steps as Case I, we can rewrite (\ref{surface_e2}) as
\begin{align*}
     &\dot{e}_i + \lambda_{1i} e_i + ( \lambda_{2i}-\frac{s_i}{\int_0^t{e_i(\tau) d\tau}})  \int_0^t{e_i(\tau)d\tau} \\
     &\qquad + \lambda_{3i} (\alpha_1 e_i + \alpha_2 \sign(e_i) e^2_i) = 0,
\end{align*}
and obtain boundedness of $e_i$ and $\int_0^t{e_i(\tau)d\tau}$. Proceeding as in Case I and combining the two cases, we obtain convergence of $\mathbf{e}$ in finite-time to the region $\Phi$ in Theorem 1, and convergence of $\int_0^t{\mathbf{e}(\tau)d\tau}$ in finite-time to the region $\bar{\Phi} = \max \{\sqrt{n}\max_i \{\frac{|e_i(0)|}{\sqrt{\lambda_{2i}}}\}, \frac{\sqrt{n} \bar{s}}{\min_i\{\lambda_{2i}\}}\}$. 
The proof is completed.
\end{proof}

\begin{remark}[Estimation of unstructured uncertainty]
The role of the adaptive law (\ref{rho}) is to estimate the uncertainty (\ref{linearbound}) without its structural knowledge. In fact, (\ref{linearbound}) has been derived from EL properties that hold independently of the specific system structure. {\color{black}It is known that nonparametric perturbations may destroy stability of classic adaptive
control \cite{robustadaptive}. As the proposed method achieves stability in the presence of unstructured state-dependent perturbations, it is "robust" in the sense of robust adaptive control.}\vspace{-0.2cm} 
\end{remark}

{\color{black}
\begin{remark}[Modularity]
To illustrate the modular nature of the proposed method, let use (\ref{surface}) and rewrite (\ref{controller}) as:
\begin{align}
\begin{split} \label{PID}
    \mathbf u(t) = & -\overbrace{\boldsymbol \Lambda \mathbf{s}_1(t)}^{\text{PID action}}  -\overbrace{\rho(t){\color{black} \sign}({\mathbf{s}(t)/\mu})}^{\text{{\color{black}Robust}-adaptive action}}  \\
    &\quad -\overbrace{\boldsymbol \Lambda \boldsymbol \lambda_3 \Delta(\mathbf{e}) -\sigma\sig^{v}(\mathbf{s}(t))}^{\text{Finite-time action}},
\end{split}
\end{align}
where {\color{black}$\mathbf{s}$ in \eqref{surface} was split to highlight} $\mathbf{s}_1 = \dot{\mathbf e} + \boldsymbol \lambda_1 \mathbf{e} + \boldsymbol \lambda_2\int_0^t \mathbf e(\tau)d\tau$.  {\color{black}Similar to the AISMC in \cite{9715173}, (\ref{PID}) allows to augment 
existing PID loops. The novelty is to include 
 a finite-time action to better control  convergence. The effectiveness of this method will be verified in Sect. \ref{Sect: Validation} with comparative experiments with PID and AISMC loops.\vspace{-0.1cm}}
\end{remark}
}


\begin{figure}[!t]
    \vspace{-0.2cm}\hspace{-0.6cm}\includegraphics[width=0.54\textwidth]{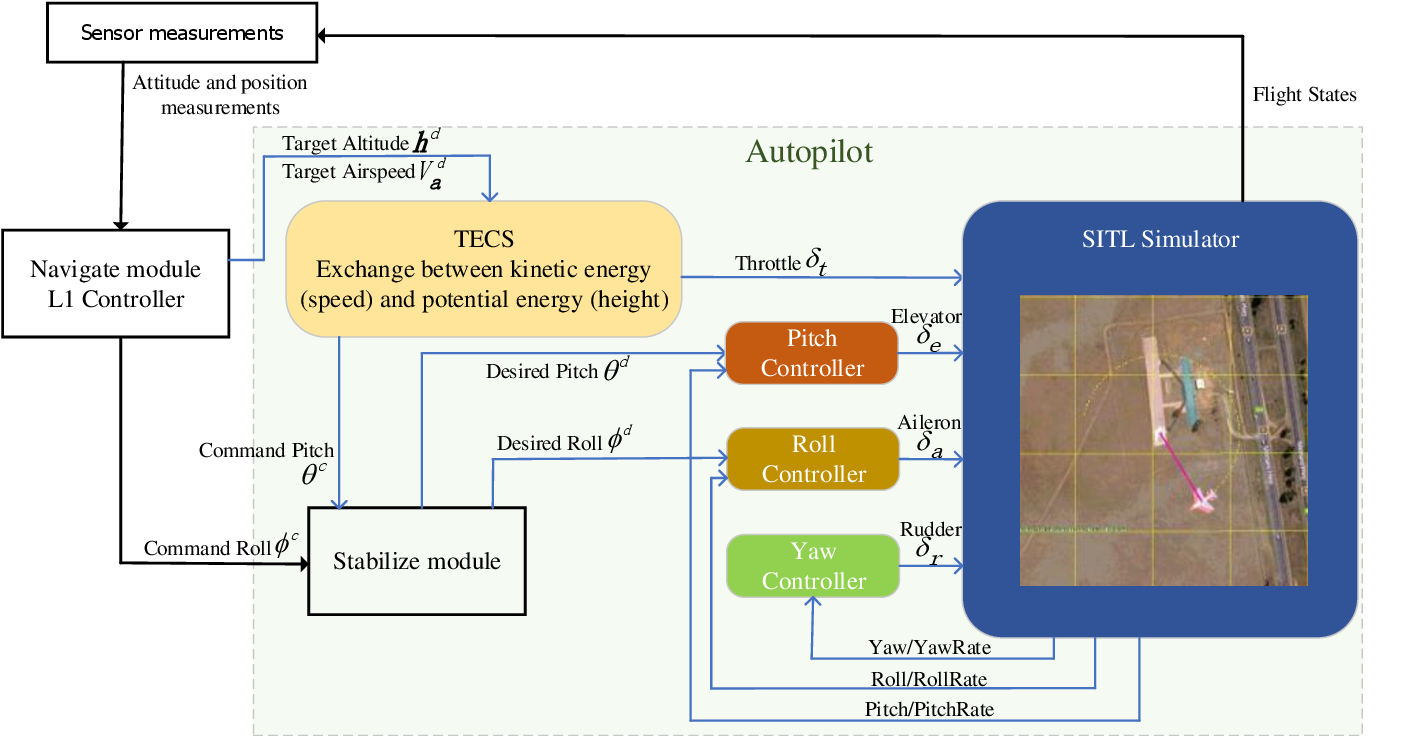}
    \caption{SITL architecture in ArduPilot: {\color{black}the Navigate module 
    sends guidance commands to the TECS and Stabilize modules; 
    roll-pitch-yaw controllers act on the corresponding servos.}\vspace{-0.2cm}}
    \label{SITL}
\end{figure}

\section{{\color{black}Validation in ArduPilot}} \label{Sect: Validation}
We validate the proposed framework into {\color{black}ArduPlane, the autopilot module of ArduPilot for fixed-wing unmanned aerial vehicles.} PID control 
and the authors' previous AISMC work  \cite{9715173} are used for comparisons, where 
\begin{align}\label{AISMC}
\begin{split}
    \mathbf u_{\rm AISMC}(t) &= -\boldsymbol \Lambda \mathbf{s}_1(t) - \rho \sat({\mathbf{s}_1(t)/\mu}), \\
    \rho(t) &= \bar{\theta}_0(t) + \bar{\theta}_1(t)||\boldsymbol \xi(t)|| + \bar{\theta}_2(t)||\boldsymbol \xi(t)||^2,  \\
    \dot{\bar{\theta}}_i(t) &= \|\mathbf{s}_1(t)\|\|\boldsymbol \xi(t)\|^i-\beta_i \bar{\theta}_i(t),
\end{split}
\end{align}
with $\mathbf{s}_1(t)$ being the PID term defined after (\ref{PID}) and $\sat(\cdot/\mu)$, $\mu>0$ the standard saturation function in $\pm 1$. 
In all the tests we report, we have tuned the gains 
to give the best possible performance for PID
while we have tuned the gains in the adaptive law \eqref{AISMC} to give the best possible performance for AISMC. To make the comparisons as fair as possible, 
\begin{itemize}
\item the same PID gains $\boldsymbol \Lambda$, $\boldsymbol \lambda_1$, $\boldsymbol \lambda_2$ are also used in all other methods;
\item the same adaptive gains $\beta_i$ in AISMC are also used in the proposed method.
\end{itemize}
By doing this, we are able to evaluate if and how much the additional terms in each strategy improve the performance. In AISMC and the proposed method, we use the term $\sat(\cdot /\mu)$ to replace the corresponding sign function, to allow continuity of the 
control action that is of practical importance. 

{\color{black}The validations adopt a software-in-the-loop (SITL)  environment, meaning that we test the actual open-source ArduPilot suite \cite{ardupilotCode}. Fig. \ref{SITL} illustrates the  SITL architecture.} For control, ArduPilot relies on a philosophy that aims to approximate the vehicle dynamics with second-order dynamics, and close them with PID loops \cite{ardupilotCode}. Such a  philosophy is supported by standard literature like \cite{UAVcraft}, which we explain as follows.

\begin{table}[!t]
\caption{Parameter selection for the proposed method}
\label{table:ParametersUAV}
\centering
    \begin{tabular}{| p{1.57cm} | p{0.48cm} | p{0.15cm} | p{0.48cm} | c | p{0.48cm} | p{0.35cm} | p{0.45cm} |}
        \hline
        Loop & $\gamma$ & $\varepsilon$ & $\lambda_3$ & $\beta_i$ & $\mu$ & $v$ & $\sigma$  \\
          & & & & \hspace{-0.1cm}{\scriptsize$ i \hspace{-0.05cm}=\hspace{-0.05cm} 0,\hspace{-0.02cm}1,\hspace{-0.02cm}2$}\hspace{-0.1cm} & & &   \\
        \hline
        pitch & $10^{-3}$ & 10 & $10^{-3}$ & 10 & 420 & 0.8 & 3.6  \\
        roll  & $10^{-3}$ & 10 & $10^{-3}$ & 160 & 1000 & 0.95 & 0.001   \\
        yaw   & $10^{-3}$ & 10 & $10^{-3}$ & 0.001 & 0.001 & 0.5 & 1.4  \\
        TECS throttle  & $10^{-3}$ & 10 & $10^{-3}$ & 0.001 & 10.4 & 0.01 & 0.04 \\
        TECS pitch  & $10^{-3}$ & 10 & $10^{-3}$ & 4 & 44 & 0.47 & 0.05  \\
        \hline
    \end{tabular}
\end{table}

\begin{table}[!b]
\caption{Tracking error costs for proposed, original PID and AISMC autopilots. The percentage degradation with respect to the proposed method is indicated.}
\label{table:UAV cost}
\vspace{-0.2cm}
\centering
    \begin{tabular}{p{1.0cm} p{0.4cm} p{0.4cm} p{0.4cm} p{0.9cm} p{0.9cm} p{1.8cm}}
        \hline
        \hline
        \multirow{2}{*}{Mass} & \multicolumn{6}{c}{Proposed autopilot} \\
        \cline{2-7}
        & Roll & Pitch & Yaw & TECS throttle & TECS pitch & Total\\
        \hline
        2 $\rightarrow$ 1kg & 1.34 & 0.99 & 5.31 & 2.75 & 2.55 & \textbf{12.94} \\
        \hline
        2kg & 1.01 & 1.01 & 0.49 & 0.51 & 0.65 & \textbf{3.67} \\
        \hline
        2 $\rightarrow$ 4kg & 0.8 & 1.07 & 3.52 & 1.73 & 2.55 & \textbf{9.67} \\
        \hline
        \hline
        \multirow{2}{*}{Mass} & \multicolumn{6}{c}{Original PID autopilot} \\
        \cline{2-7}
        & Roll & Pitch & Yaw & TECS throttle & TECS pitch & Total  $\qquad \ \qquad \qquad$ $\qquad \  $ $\left.\right.\left.\right.\left.\right.\left.\right.\left.\right.$ \\ 
        \hline
        2 $\rightarrow$ 1kg & 1.38 & 0.88 & 5.68 & 10.33 & 3.09 & 21.36  (+65.1\%) \\
        \hline
        2kg & 1 & 1 & 1 & 1 & 1 & 5.0 $ \ \  $ (+36.2\%) \\
        \hline
        2 $\rightarrow$ 4kg & 0.72 & 1.3 & 4.0 & 5.43 & 5.9 & 17.35  (+79.4\%) \\
        \hline
        \hline
        \multirow{2}{*}{Mass} & \multicolumn{6}{c}{AISMC autopilot} \\
        \cline{2-7}
        & Roll & Pitch & Yaw & TECS throttle & TECS pitch & Total  $\qquad \ \qquad \qquad$ $\qquad \  $ $\left.\right.\left.\right.\left.\right.\left.\right.\left.\right.$ \\ 
        \hline
        2 $\rightarrow$ 1kg & 1.47 & 0.94 & 5.44 & 4.94 & 2.97 & 15.76  (+21.8\%) \\
        \hline
        2kg & 1 & 1 & 0.49 & 0.73 & 0.92 & 4.14 $\ $ (+12.8\%) \\
        \hline
        2 $\rightarrow$ 4kg & 0.78 & 1.17 & 3.6 & 2.55 & 3.9 & 12.0 $\ $ (+24.1\%) \\
        \hline
        \hline
    \end{tabular}
\end{table}

\begin{figure*}[!t]
 \centering
\begin{subfigure}[t]{.48\linewidth}    \includegraphics[width=1.05\textwidth]{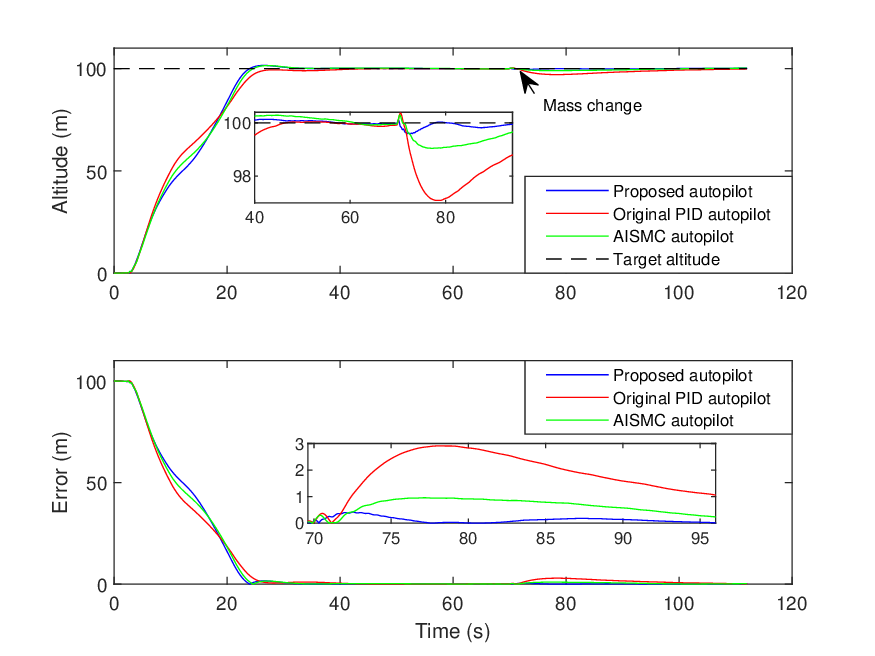}
    \centering
    \caption{{\color{black} Altitude and altitude error norm. The proposed solution \newline has negligible altitude drop after mass change.}}
    \label{fig:height1}
\end{subfigure}
\begin{subfigure}[t]{.48\linewidth}	
\centering
    \includegraphics[width=1.05\textwidth]{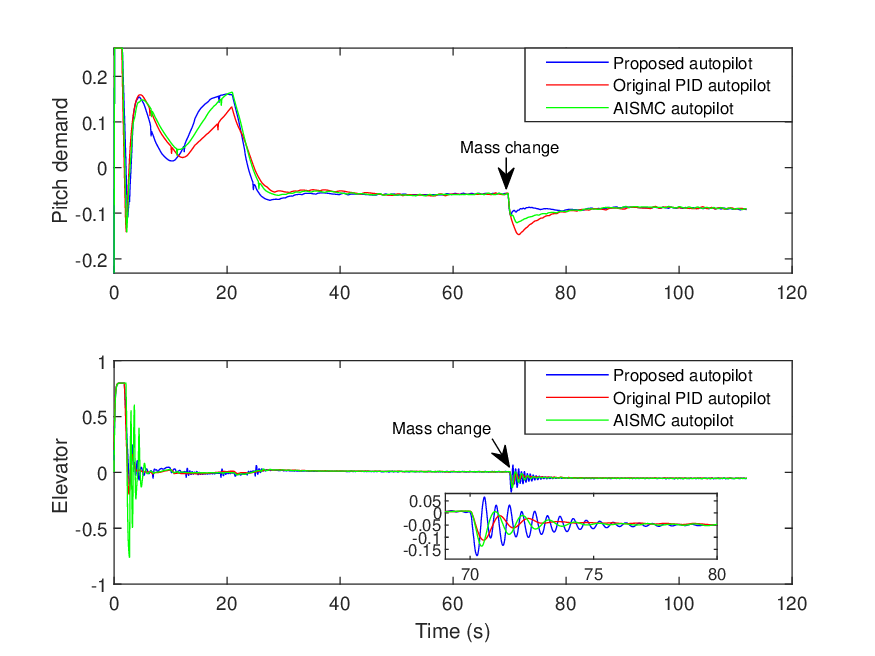}
    \centering
    \caption{{\color{black} TECS Pitch demand and elevator. The proposed solution  is more \newline reactive to mass change,  with elevator input in reasonable range.}}
    \label{fig:pitch_elevator1}
\end{subfigure}
\begin{subfigure}[t]{.48\linewidth}	
\centering
\end{subfigure}
\begin{subfigure}[t]{.48\linewidth}
\includegraphics[width=1.05\textwidth]{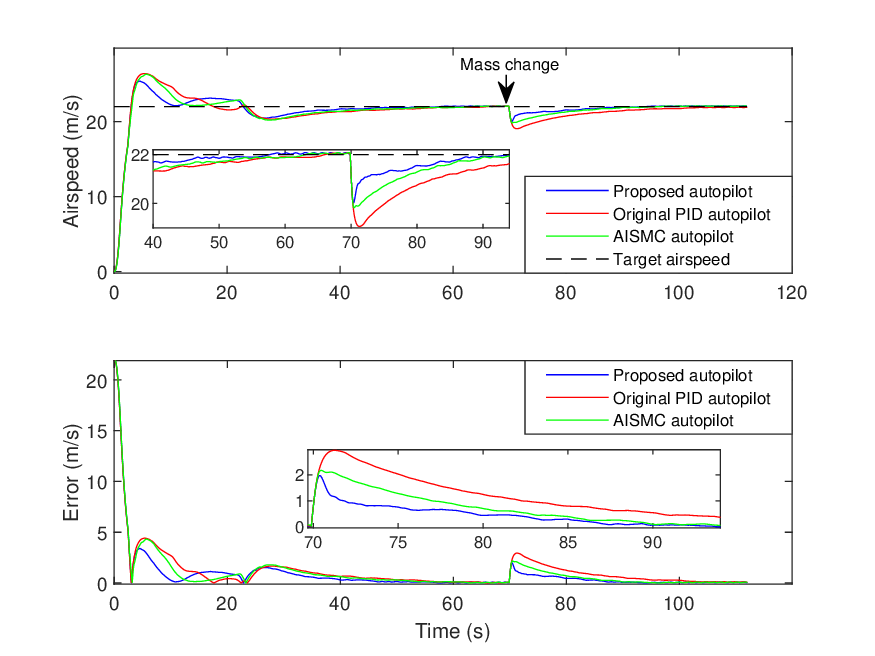}
    \centering
    \caption{Airspeed and airspeed error norm. The proposed solution has \newline smaller overshoot with faster convergence after mass change.}
    \label{fig:Va1}
\end{subfigure}
\begin{subfigure}[t]{.48\linewidth}	
\centering 
    \includegraphics[width=1.05\textwidth]{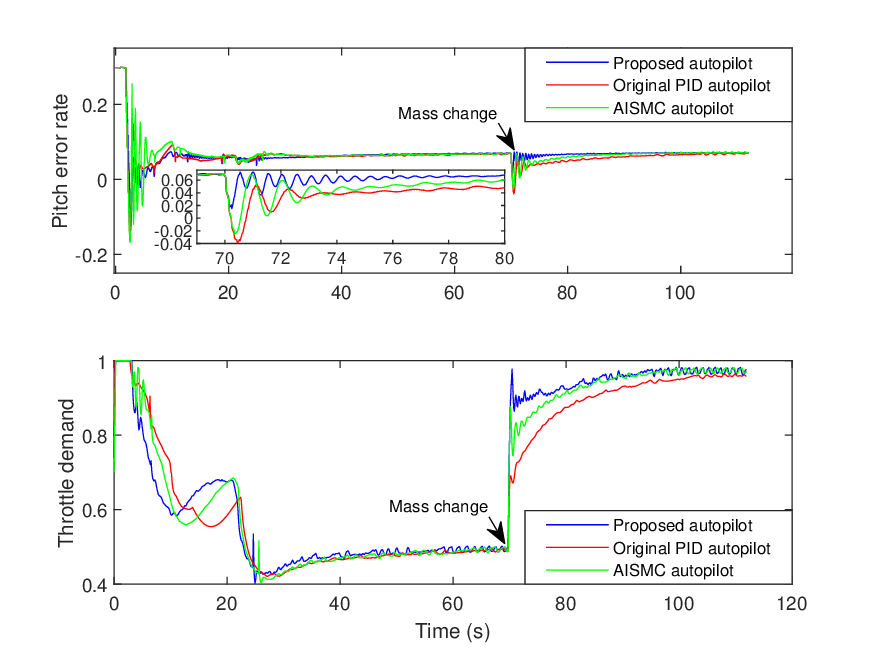}
    \centering
    \caption{{\color{black} Pitch error rate and TECS throttle. The proposed solution is \newline more  reactive to mass change.}}
    \label{fig:throttle_pitcherror1}
\end{subfigure}
\caption{Mass change 2.0kg $\rightarrow$ 1.0kg: comparison of proposed, original PID and AISMC \cite{9715173} autopilots.\vspace{-0.1cm}}
\label{fig:uav1}
\end{figure*}

It is well known that the 6-DOF fixed-wing dynamics 
are coupled and complex \cite[Chapt. 3]{UAVcraft}. So, several literature suggests to simplify the control design via low-order models (refer to \cite[Chapt. 5]{UAVcraft} for a details). For compactness, let us use the transfer function notation to describe the relation between the elevator deflection $\delta_e$ and the pitch angle $\theta$, between the pitch angle $\theta$ and the altitude $h$, and between the throttle $\delta_t$ and pitch angle $\theta$ to the airspeed $V_a$. These are the main components of the longitudinal dynamics: 
\begin{align} \label{longitudinal}
    &\text{pitch} \ \ \ \ \ \  \theta(s) =  \frac{a_{\theta3}}{s^2+a_{\theta1}s+a_{\theta2}}(\delta_e(s)+\frac{1}{a_{\theta3}}d_{\theta2}(s)), \nonumber \\
    &\text{altitude} \ \ \  h(s) = \frac{V_a}{s}(\theta(s) + \frac{1}{V_a}d_h(s)), \\
    &\text{airspeed} \ \  \bar{V}_a(s) = \frac{1}{s+a_{v_1}}(a_{v_2}\bar{\delta}_t(s) - a_{v_3}\bar{\theta}(s) + d_v(s)), \nonumber
\end{align}
where $a_{\theta1} \triangleq -\frac{\rho_a V_a^2 cS}{2J_y}C_{m_q} \frac{c}{2V_a},$ 
\begin{align*}
    a_{\theta2} &\triangleq -\frac{\rho_a V_a^2cS}{2J_y}C_{m_\alpha}, \quad
    a_{\theta3} \triangleq \frac{\rho_a V_a^2 cS}{2J_y}C_{m_{\delta_e}}, \\
    d_{\theta2} &\triangleq \hspace{-0.08cm} \Gamma_6(\hspace{-0.04cm}r^2 \hspace{-0.18cm}-\hspace{-0.08cm}p^2) \hspace{-0.09cm}+\hspace{-0.09cm} \Gamma_5pr \hspace{-0.09cm}+\hspace{-0.09cm}  \frac{\rho_a V_a^2 cS}{2J_y}[C_{m_0}\hspace{-0.15cm}-\hspace{-0.1cm} C_{m_\alpha}\hspace{-0.1cm} \gamma \hspace{-0.09cm}-\hspace{-0.09cm} C_{m_q} \frac{c}{2V_a}d_{\theta1}] \hspace{-0.09cm}+\hspace{-0.09cm} \dot{d}_{\theta1}, \\
    d_h &\triangleq (u \sin{\theta} - V_a \theta) - v \sin{\phi} \cos{\theta} -w \cos{\phi} \cos{\theta}, \\
    \bar{V}_a &\triangleq V_a - V_a^* \ \text{is the deviation of } V_a \ \text{from trim } V_a^*, \\
    \bar{\theta} &\triangleq \theta - \theta^* \hspace{-0.03cm} \ \quad \text{is the deviation of } \theta \ \text{from trim } \theta^*, \\
    \bar{\delta}_t &\triangleq \delta_t - \delta_t^* \hspace{-0.01cm} \quad \text{is the deviation of } \delta_t \ \text{from trim } \delta_t^*, \\
    a_{v_1} &\triangleq \frac{\rho V_a^*S}{m_u}(C_{D_0}+C_{D_{\alpha^*}}+C_{D_{\delta_e}}\delta_e^*) + \frac{\rho S_{\text{prop}}}{m_u} C_{\text{prop}}V_a^*, \\
    a_{v_2} &\triangleq \frac{\rho S_{\text{prop}}}{m_u} C_{\text{prop}} k_{\text{motor}}^2 \delta_t^*, \quad
    a_{v_3} \triangleq g \cos{(\theta^* - \chi^*)},
\end{align*}
where $\phi$ is the Euler roll angle, $p$, $r$ are the roll rate and yaw rate in body frame, $\gamma$ is flight path angle, $\chi^*$ is the course angle trim and the terms $C_{m_q},  C_{m_\alpha}, C_{m_{\delta_e}}, C_{m_0}, C_{D_0}, C_{D_{\alpha^*}}$,  $C_{D_{\delta_e}}$ are Taylor approximations of aerodynamic coefficients. We do not report the lateral dynamics for compactness,  the interested reader is referred to \cite[Chapt. 5]{UAVcraft}.  It is worth remarking that most coefficients in (\ref{longitudinal}) are uncertain and {\color{black}difficult to identify, as they might  change with time and with the operating conditions.} Additionally,   $d_{\theta2}$, $d_h$, $d_v$ represent complex state-dependent disturbances coming from the coupled dynamics. {\color{black}These issues make model-based or structure-based autopilots hard to implement in practice: meanwhile,  {\color{black}non-adaptive PID loops may guarantee practical stability, but their guarantees in the presence of state-dependent uncertainty may be conservative}, as shown later in our tests.}

{\color{black}Since ArduPilot relies on closing the vehicle dynamics with PID loops \cite{ardupilotCode}, by recalling 
\eqref{PID}, we have the opportunity to test the proposed  method by augmenting the original loops of ArduPilot.} 
Indeed, the attitude control modules in Fig. \ref{SITL} can either contain the original ArduPilot autopilot, or a user-designed autopilot (e.g., the AISMC \cite{9715173} or the approach proposed in this work). As the original ArduPilot consists of a family of PID loops, one can choose  $\boldsymbol \Lambda$, $\boldsymbol \lambda_1$, $\boldsymbol \lambda_2$ in (\ref{controller}) as the PID gains set in the ArduPilot code \cite{ardupilotCode}. 
The other parameters for the proposed approach are listed in Table \ref{table:ParametersUAV}. 

\begin{figure*}[!t]
 \centering
\begin{subfigure}[t]{.48\linewidth}    \includegraphics[width=1.05\textwidth]{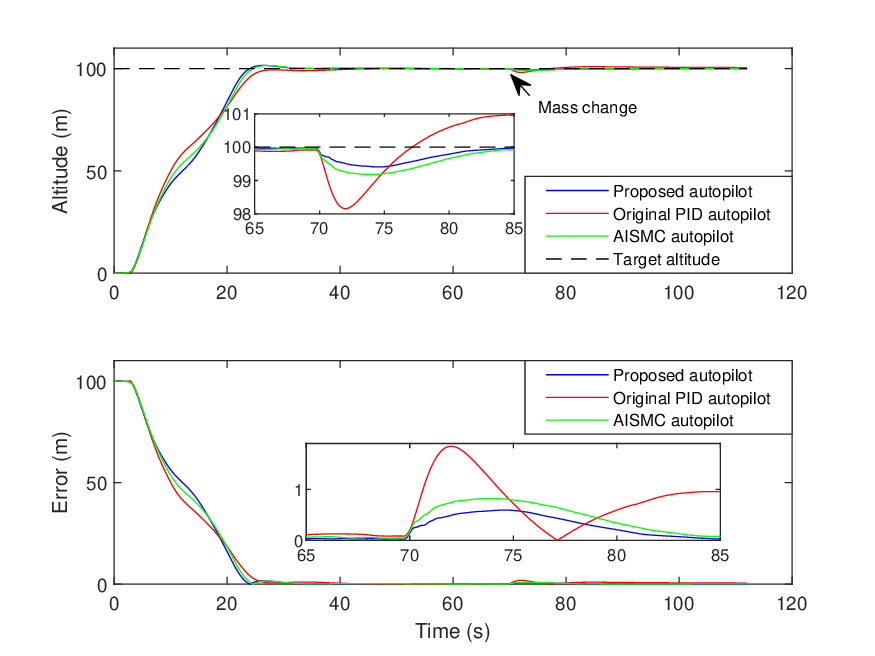}
    \centering
    \caption{{\color{black} Altitude and altitude error norm. The proposed solution \newline has negligible altitude drop after mass change.}}
    \label{fig:height4}
\end{subfigure}
\begin{subfigure}[t]{.48\linewidth}	
\centering
    \includegraphics[width=1.05\textwidth]{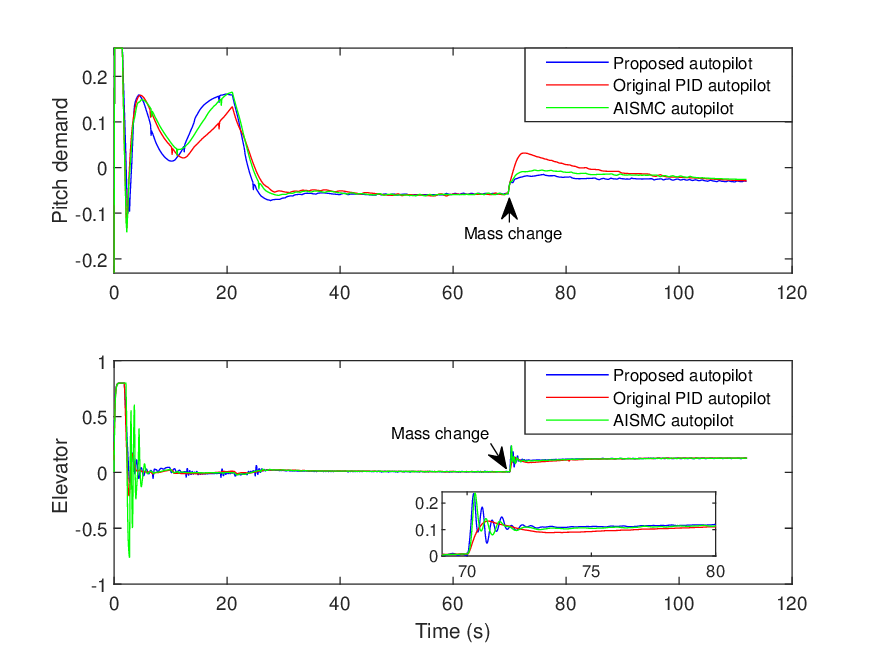}
    \centering
    \caption{{\color{black} TECS pitch demand and elevator. The proposed solution is more \newline  reactive to mass change, with elevator input in reasonable range.}}
    \label{fig:pitch_elevator4}
\end{subfigure}
\begin{subfigure}[t]{.48\linewidth}	
\centering
\end{subfigure}
\begin{subfigure}[t]{.48\linewidth}
\includegraphics[width=1.05\textwidth]{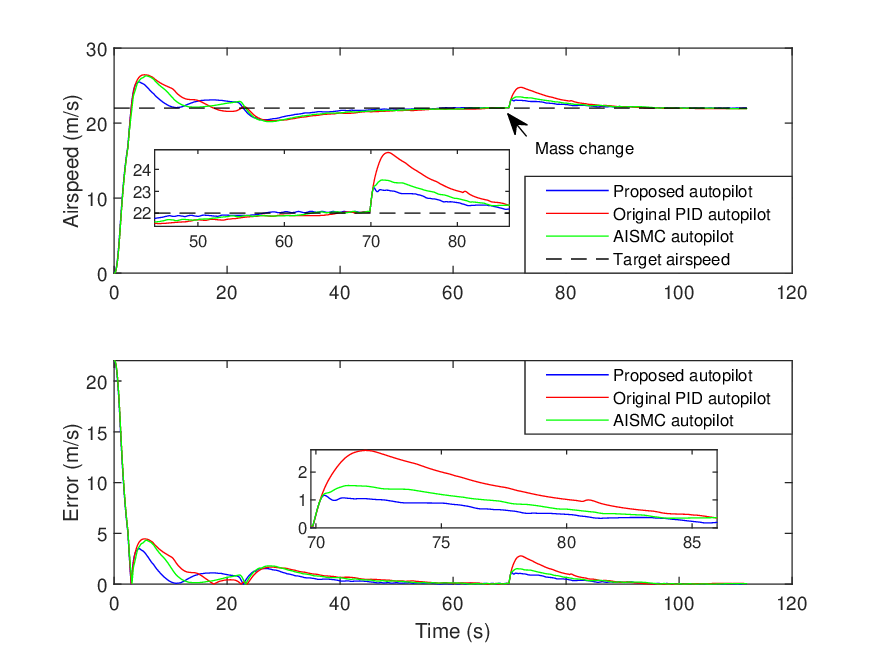}
    \centering
    \caption{Airspeed and airspeed error norm. The proposed solution has \newline smaller overshoot with faster convergence after mass change.}
    \label{fig:Va4}
\end{subfigure}
\begin{subfigure}[t]{.48\linewidth}	
\centering 
    \includegraphics[width=1.05\textwidth]{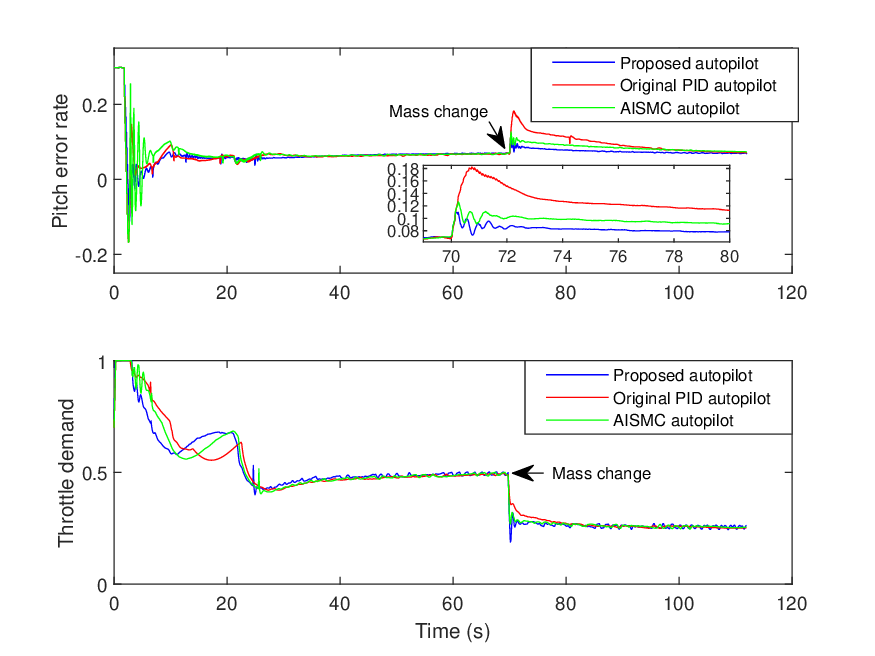}
    \centering
    \caption{{\color{black} Pitch error rate and TECS throttle. The proposed solution  is \newline more  reactive to mass change.}}
    \label{fig:throttle_pitcherror4}
\end{subfigure}
\caption{Mass change 2.0kg $\rightarrow$ 4.0kg: comparison of proposed, original PID and AISMC \cite{9715173} autopilots.\vspace{-0.1cm}}
\label{fig:uav4}
\end{figure*}

To evaluate the robustness of the proposed autopilot under large uncertainty, the mass of the UAV will be changed during the flight (which can represent some change in the payload). The initial mass is $2$kg, and it can change to half mass ($1$kg) or double mass ($4$kg) during flight. Neither the mass nor the mass change is known a priori, which allows to test how different autopilots cope with this uncertainty. The flight includes a take-off phase, a cruising phase (orbiting  constant altitude and constant airspeed), and a landing phase. Note that the mass change occurs during the cruising phase. Table \ref{table:UAV cost} reports the tracking error costs for the different autopilots under different mass conditions. Five control loops are reported, representing the main loops of a fixed-wing UAV: roll, pitch, yaw, TECS throttle and TECS pitch. The term TECS means total energy control system, i.e. the control loops for altitude (potential energy) and speed (kinetic energy). The cost accounts for the  tracking errors with respect to the desired roll, pitch, yaw, potential and kinetic energy. The performance of the original PID autopilot with constant 2kg mass is used as a normalizing factor, so that the cost is 5.0 in this scenario. The table shows that the tracking error costs of the proposed autopilot are smaller than other autopilots under both constant mass scenario and changing mass scenario: the proposed autopilot  overcomes the original autopilot by more than 36.2\% for the original PID autpilot and by more than 12.8\% for the AISMC autopilot. The enhanced performance is especially evident in the changing mass scenario.

To visualize the performance in the changing mass scenario,  relevant flight variables are visualized in Figs. \ref{fig:uav1} and \ref{fig:uav4} for the proposed, the original PID and the AISMC autopilots. As shown in Figs. \ref{fig:height1} and \ref{fig:height4}, the proposed autopilot has negligible altitude drop after mass change. Figs. \ref{fig:pitch_elevator1} and \ref{fig:pitch_elevator4} show the faster reaction in pitch demand, with reasonable elevator input. Figs. \ref{fig:Va1} and \ref{fig:Va4} show that the airspeed of the proposed autopilot has smaller overshoot and converges faster. As shown in Figs. \ref{fig:throttle_pitcherror1} and \ref{fig:throttle_pitcherror4}, the 
throttle for the proposed solution is more reactive in response to the mass change.

To conclude, the tests show that the proposed method overcomes the other ones in terms of robustness and adaptation, {\color{black}suggesting  more effectiveness in tackling state-dependent uncertainty typically arising in autopilot
applications.} {\color{black}Additional comparisons can be found in the report \cite{TechReportandCode}.}

\section{Conclusions and future work}
This paper explored a novel boundary-layer Euler-Lagrange (EL) control method not relying on structural knowledge of the system dynamics. 
{\color{black}
Design considerations and software-in-the-loop tests have shown the capability of the proposed solution to be integrated in existing autopilot loops, with enhanced robustness and adaptation performance}. Future work will aim to test several modules in open-source autopilots meant for vehicles with different structure (copters, vessels, rovers, etc.). A preliminary study in this sense is \cite{9993292}. {\color{black}Another unsettled problem in the field of unstructured uncertainties is to consider unmatched/unactuated terms.} 

\bibliographystyle{IEEEtran}
\bibliography{references}

@ARTICLE{9165001,
  author={Verginis, Christos K. and Dimarogonas, Dimos V.},
  journal={IEEE Transactions on Automatic Control}, 
  title={Asymptotic Tracking of Second-Order Nonsmooth Feedback Stabilizable Unknown Systems With Prescribed Transient Response}, 
  year={2021},
  volume={66},
  number={7},
  pages={3296-3302},
  doi={10.1109/TAC.2020.3015785}}

@INPROCEEDINGS{9030217,
  author={Verginis, Christos K. and Dimarogonas, Dimos V.},
  booktitle={2019 IEEE 58th Conference on Decision and Control (CDC)}, 
  title={Asymptotic Stability of Uncertain Lagrangian Systems with Prescribed Transient Response}, 
  year={2019},
  volume={},
  number={},
  pages={7037-7042},
  doi={10.1109/CDC40024.2019.9030217}}

@INPROCEEDINGS{9993292,
  author={Li, Peng and Liu, Di and Xia, Xin and Baldi, Simone},
  booktitle={2022 IEEE 61st Conference on Decision and Control (CDC)}, 
  title={Embedding Adaptive Features in the ArduPilot Control Architecture for Unmanned Aerial Vehicles}, 
  year={2022},
  volume={},
  number={},
  pages={3773-3780},
  doi={10.1109/CDC51059.2022.9993292}}

@article{https://doi.org/10.1002/rnc.3562,
author = {Safonov, Michael G.},
title = {Robust control: Fooled by assumptions},
journal = {International Journal of Robust and Nonlinear Control},
volume = {28},
number = {12},
pages = {3667-3677},
doi = {https://doi.org/10.1002/rnc.3562},
year = {2018}
}

@ARTICLE{9525201,
  author={Cruz-Ancona, Christopher D. and Estrada, Manuel A. and Fridman, Leonid},
  journal={IEEE Transactions on Automatic Control}, 
  title={Barrier Function-Based Adaptive Lyapunov Redesign for Systems Without A Priori Bounded Perturbations}, 
  year={2022},
  volume={67},
  number={8},
  pages={3851-3862},
  doi={10.1109/TAC.2021.3107453}}

@ARTICLE{9795675,
  author={Haimovich, Hernan and Fridman, Leonid and Moreno, Jaime A.},
  journal={IEEE Transactions on Automatic Control}, 
  title={Generalized Super-Twisting for Control Under Time- and State-Dependent Perturbations: Breaking the Algebraic Loop}, 
  year={2022},
  volume={67},
  number={10},
  pages={5646-5652},
  doi={10.1109/TAC.2022.3183039}}

@ARTICLE{9271817,
  author={Yu, Xinghuo and Feng, Yong and Man, Zhihong},
  journal={IEEE Open Journal of the Industrial Electronics Society}, 
  title={Terminal Sliding Mode Control – An Overview}, 
  year={2021},
  volume={2},
  number={},
  pages={36-52},
  doi={10.1109/OJIES.2020.3040412}}

@article{doi:10.1080/00207179.2010.501385,
author = { F.   Plestan  and  Y.   Shtessel  and  V.   Brégeault  and  A.   Poznyak },
title = {New methodologies for adaptive sliding mode control},
journal = {International Journal of Control},
volume = {83},
number = {9},
pages = {1907-1919},
year  = {2010},
doi = {10.1080/00207179.2010.501385},
}

@article{doi:10.1080/00207179.2016.1138241,
author = {Gian Paolo Incremona and Michele Cucuzzella and Antonella Ferrara},
title = {Adaptive suboptimal second-order sliding mode control for microgrids},
journal = {International Journal of Control},
volume = {89},
number = {9},
pages = {1849-1867},
year  = {2016},
doi = {10.1080/00207179.2016.1138241},
}

@ARTICLE{9629337,  author={Dong, Haotian and Xi, Junqiang},  journal={IEEE Transactions on Vehicular Technology},   title={Model Predictive Longitudinal Motion Control for the Unmanned Ground Vehicle With a Trajectory Tracking Model},   year={2022},  volume={71},  number={2},  pages={1397-1410},  doi={10.1109/TVT.2021.3131314}}

@article{VONELLEN102750,
title = {Shared human–robot path following control of an unmanned ground vehicle},
journal = {Mechatronics},
volume = {83},
pages = {102750},
year = {2022},
issn = {0957-4158},
doi = {https://doi.org/10.1016/j.mechatronics.2022.102750},
author = {Karl D. {von Ellenrieder} and Stephen C. Licht and Roberto Belotti and Helen C. Henninger}
}

@article{ZUO2015305,
title = {Nonsingular fixed-time consensus tracking for second-order multi-agent networks},
journal = {Automatica},
volume = {54},
pages = {305-309},
year = {2015},
issn = {0005-1098},
doi = {https://doi.org/10.1016/j.automatica.2015.01.021},
author = {Zongyu Zuo},
}

@ARTICLE{7959105,  author={Wang, Fang and Chen, Bing and Lin, Chong and Zhang, Jing and Meng, Xinzhu},  journal={IEEE Transactions on Cybernetics},   title={Adaptive Neural Network Finite-Time Output Feedback Control of Quantized Nonlinear Systems},   year={2018},  volume={48},  number={6},  pages={1839-1848},  doi={10.1109/TCYB.2017.2715980}}

@article{2018Adaptive,
  title={Adaptive fixed-time bipartite tracking consensus control for unknown nonlinear multi-agent systems: An information classification mechanism},
  author={ Yang, H.  and  Ye, D. },
  journal={Information Sciences},
  pages={238-254},
  year={2018},
}

@article{2020Adaptive,
  title={Adaptive Tracking Control of an Electronic Throttle Valve Based on Recursive Terminal Sliding Mode},
  author={ Hu, Youhao  and  Wang, Hai  and  He, Shuping  and  Zheng, Jinchuan  and  Man, Zhihong },
  journal={IEEE Transactions on Vehicular Technology},
  volume={1},
  number={1},
  pages={1-11},
  year={2020},
}

@ARTICLE{8693689,  author={Santoso, Fendy and Garratt, Matthew A. and Anavatti, Sreenatha G.},  journal={IEEE Transactions on Systems, Man, and Cybernetics: Systems},   title={Hybrid {PD}-Fuzzy and {PD} Controllers for Trajectory Tracking of a Quadrotor Unmanned Aerial Vehicle: Autopilot Designs and Real-Time Flight Tests},   year={2021},  volume={51},  number={3},  pages={1817-1829},  doi={10.1109/TSMC.2019.2906320}}

@ARTICLE{9703092,  author={Baldi, Simone and Roy, Spandan and Yang, Kang and Liu, Di},  journal={IEEE/ASME Transactions on Mechatronics},   title={An Underactuated Control System Design for Adaptive Autopilot of Fixed-Wing Drones},   year={2022},  volume={},  number={},  pages={1-12},  doi={10.1109/TMECH.2022.3144459}}

@INPROCEEDINGS{9589106,  author={Li, Peng and Liu, Di and Baldi, Simone},  booktitle={IECON 2021 – 47th Annual Conference of the IEEE Industrial Electronics Society},   title={Plug-and-play adaptation in autopilot architectures for unmanned aerial vehicles},   year={2021},  volume={},  number={},  pages={1-6},  doi={10.1109/IECON48115.2021.9589106}}

@article{2006Adaptive,
  title={Adaptive fuzzy control for the ship steering problem},
  author={ Gerasimos Rigatos  and  Spyros Tzafestas },
  journal={Mechatronics},
  volume={16},
  number={ 8},
  pages={479-489},
  year={2006},
}

@ARTICLE{9369849,  author={Rayguru, Madan Mohan and Mohan, Rajesh Elara and Parween, Rizuwana and Yi, Lim and Le, Anh Vu and Roy, Spandan},  journal={IEEE/ASME Transactions on Mechatronics},   title={An Output Feedback Based Robust Saturated Controller Design for Pavement Sweeping Self-Reconfigurable Robot},   year={2021},  volume={26},  number={3},  pages={1236-1247},  doi={10.1109/TMECH.2021.3063886}}

@ARTICLE{8438315,  author={Kürkçü, Burak and Kasnakoğlu, Coşku},  journal={IEEE Transactions on Control Systems Technology},   title={Robust Autopilot Design Based on a Disturbance/Uncertainty/Coupling Estimator},   year={2019},  volume={27},  number={6},  pages={2622-2629},  doi={10.1109/TCST.2018.2859179}}

@ARTICLE{8304792,  author={Fu, Changhong and Sarabakha, Andriy and Kayacan, Erdal and Wagner, Christian and John, Robert and Garibaldi, Jonathan M.},  journal={IEEE/ASME Transactions on Mechatronics},   title={Input Uncertainty Sensitivity Enhanced Nonsingleton Fuzzy Logic Controllers for Long-Term Navigation of Quadrotor {UAVs}},   year={2018},  volume={23},  number={2},  pages={725-734},  doi={10.1109/TMECH.2018.2810947}}

@ARTICLE{9570132,  author={Yang, Zeyu and Huang, Jin and Yin, Hui and Yang, Diange and Zhong, Zhihua},  journal={IEEE Transactions on Intelligent Transportation Systems},   title={Path Tracking Control for Underactuated Vehicles With Matched-Mismatched Uncertainties: An Uncertainty Decomposition Based Constraint-Following Approach},   year={2021},  volume={},  number={},  pages={1-14},  doi={10.1109/TITS.2021.3118375}}

@ARTICLE{9325954,  author={Roy, Spandan and Baldi, Simone and Li, Peng and Narayanan, Viswa},  journal={IEEE/ASME Transactions on Mechatronics},   title={Artificial-Delay Adaptive Control for Under-actuated {Euler-Lagrange} Robotics},   year={2021},  volume={},  number={},  pages={1-1},  doi={10.1109/TMECH.2021.3052068}}

@article{2016Adaptive,
  title={Adaptive Continuous Higher Order Sliding Mode Control},
  author={ Edwards, Christopher  and  Shtessel, Yuri B. },
  journal={Automatica},
  volume={65},
  pages={183-190},
  year={2016},
}

@ARTICLE{7859476,  author={Nuño, Emmanuel and Ortega, Romeo},  journal={IEEE Transactions on Control Systems Technology},   title={Achieving Consensus of {Euler–Lagrange} Agents With Interconnecting Delays and Without Velocity Measurements via Passivity-Based Control},   year={2018},  volume={26},  number={1},  pages={222-232},  doi={10.1109/TCST.2017.2661822}}

@ARTICLE{9570130,  author={Yu, Li and He, Guang and Wang, Xiangke and Zhao, Shulong},  journal={IEEE Transactions on Industrial Electronics},   title={Robust fixed-time sliding mode attitude control of tilt tri-rotor {UAV} in Helicopter mode},   year={2021},  volume={},  number={},  pages={1-1},  doi={10.1109/TIE.2021.3118556}}

@ARTICLE{9512516,  author={Wan, Lucas and Pan, Ya-Jun and Shen, Henghua},  journal={IEEE/ASME Transactions on Mechatronics},   title={Improving Synchronization Performance of Multiple {Euler-Lagrange} Systems using Non-Singular Terminal Sliding Mode Control with Fuzzy Logic},   year={2021},  volume={},  number={},  pages={1-1},  doi={10.1109/TMECH.2021.3104504}}

@ARTICLE{8435993,  author={Ma, Carlos and Lam, James and Lewis, Frank L.},  journal={IEEE Transactions on Control Systems Technology},   title={Trajectory Regulating Model Reference Adaptive Controller for Robotic Systems},   year={2019},  volume={27},  number={6},  pages={2749-2756},  doi={10.1109/TCST.2018.2858203}}

@ARTICLE{9721593,  author={Souanef, Toufik},  journal={IEEE Transactions on Aerospace and Electronic Systems},   title={$\mathcal{L}_1$ Adaptive Path-Following of Small Fixed-wing Unmanned Aerial Vehicles in Wind},   year={2022},  volume={},  number={},  pages={1-1},  doi={10.1109/TAES.2022.3153758}}

@ARTICLE{9731672,
  author={Shao, Ke and Zheng, Jinchuan and Tang, Rongchuan and Li, Xiu and Man, Zhihong and Liang, Bin},
  journal={IEEE/ASME Transactions on Mechatronics}, 
  title={Barrier Function Based Adaptive Sliding Mode Control for Uncertain Systems With Input Saturation}, 
  year={2022},
  volume={27},
  number={6},
  pages={4258-4268},
  doi={10.1109/TMECH.2022.3153670}}

@ARTICLE{9001188,  author={S. {Roy} and J. {Lee} and S. {Baldi}},  journal={IEEE Transactions on Control Systems Technology},   title={A New Adaptive-Robust Design for Time Delay Control Under State-Dependent Stability Condition},   year={2021},  volume={29},  number={1},  pages={420-427},  doi={10.1109/TCST.2020.2969129}}

@article{2020Ona,
  title={On adaptive sliding mode control without a priori bounded uncertainty},
  author={ Roy, Spandan  and  Baldi, Simone  and  Fridman, Leonid M. },
  journal={Automatica},
  volume={111},
  pages={108650},
  year={2020},
}

@ARTICLE{7331625,  author={Wang, Ning and Qian, Chunjiang and Sun, Jing-Chao and Liu, Yan-Cheng},  journal={IEEE Transactions on Control Systems Technology},   title={Adaptive Robust Finite-Time Trajectory Tracking Control of Fully Actuated Marine Surface Vehicles},   year={2016},  volume={24},  number={4},  pages={1454-1462},  doi={10.1109/TCST.2015.2496585}}

@ARTICLE{9492051,  author={Mofid, Omid and Mobayen, Saleh},  journal={IEEE Transactions on Aerospace and Electronic Systems},   title={Adaptive finite-time back-stepping global sliding mode tracker of quad-rotor {UAVs} under model uncertainty, wind perturbation and input saturation},   year={2021},  volume={},  number={},  pages={1-1},  doi={10.1109/TAES.2021.3098168}}

@ARTICLE{9715173,
  author={Li, Peng and Liu, Di and Baldi, Simone},
  journal={IEEE/ASME Transactions on Mechatronics}, 
  title={Adaptive Integral Sliding Mode Control in the Presence of State-Dependent Uncertainty}, 
  year={2022},
  volume={27},
  number={5},
  pages={3885-3895},
  doi={10.1109/TMECH.2022.3145910}}

@book{robustadaptive,
author="Petros Ioannou and Jing Sun",
title="Robust Adaptive Control",
year="2012",
publisher="Dover Publications"
}

@Inbook{UAVcraft,
author="Randal W. Beard and Timothy W. Mclain",
title="Small Unmanned Aircraft Theory and Practice",
bookTitle="Small Unmanned Aircraft Theory and Practice",
year="2012",
publisher="Princeton University Press"
}

@misc{ardupilotCode,
title = {Open source for ArduPilot open source autopilot},
howpublished = {Online},
note = {\url{https://github.com/ArduPilot/ardupilot}}
}

@misc{TechReportandCode,
title = {Technical report on boundary-layer control with unstructured
uncertainties},
howpublished = {Online},
note = {\url{https://github.com/Friend-Peng/Adaptive-ArduPilot-Autopilot/tree/Boundary-layer-control-with-application-to-adaptive-autopilots}}
}

\end{document}